\documentclass[a4paper,11pt]{article}
\usepackage{amssymb,amsmath,amsthm}
\usepackage{marvosym}
\usepackage{graphicx}

\usepackage{todonotes, comment} 

\usepackage{amsfonts}
\usepackage{latexsym}
\usepackage{color}
\usepackage[english]{babel}
\usepackage{hyperref}

\usepackage{float}

\numberwithin{equation}{section}

\newtheorem{lemma}{Lemma}[section]
\newtheorem{theorem}[lemma]{Theorem}

\newtheorem{corollary}[lemma]{Corollary}
\newtheorem{remark}[lemma]{Remark}

\newcommand{\cK}{{\mathcal K}}

\newcommand{\rep}[2]{\underbrace{#1,\ldots,#1}_{#2\text{ times}}}

\newcommand{\R}{\mathbb R}

\newcommand{\N}{\mathbb N}

\newcommand{\Tr}{\operatorname{Tr}}

\definecolor{awesome}{rgb}{1.0, 0.13, 0.32}
\definecolor{darkgr}{rgb}{0.0, 0.62, 0.42}
\definecolor{cyan}{rgb}{0.0, 0.72, 0.92}

\def\cH{{\mathcal{H}}}

\def\cI{{\mathcal{I}}}
\def\QS{{\mathcal{Q}\kern-0.3pt\mathcal{S}}}

\def\cE{{\mathcal{E}}}

\def\cO{{\mathcal{O}}}

\def\cR{{\mathcal{R}}}

\def\Fdue{\null\,\null_2\kern-1pt F_{1}\null}

\begin{document}

\title{Long time stability of low--frequency  packets of modes at finite
  total energy
  }

\date{}


\author{ Dario Bambusi\footnote{Dipartimento di Matematica ``Federigo Enriques'', Universit\`a degli Studi di Milano, Via Saldini 50, 20133
Milano, Italy 
 \textit{Email: } \texttt{dario.bambusi@unimi.it}}, and
Antonio Ponno\footnote{Dipartimento di Matematica "T. Levi-Civita", Universit\`a degli Studi di Padova, Via Trieste 63, 35121 Padova, Italy
 \textit{Email: } \texttt{ponno@math.unipd.it}}}

\maketitle

\begin{abstract}
  We consider the FPUT $\alpha$-model with $N$ particles and 
  initial data with small total energy $E$, independent of $N$. In particular, 
  we prove that if the energy is initially given to Fourier modes of
  wave number
  $\pi k/N\leq \theta_*$, then it does not flow to modes of wave number larger than
  $\theta_*$ up to long times of order $E^{-(1-\epsilon/2)}$, $\theta_*$ and $\epsilon$ being small numbers. The proof relies on the closeness 
  of the Toda lattice to the FPUT system. Unlike
  previous approaches, we do not make use of the action-angle variables of the
  Toda system, but we exploit the Flaschka integrals of the latter. Actually, we find a
  suitable recombination of the Flaschka integrals whose quadratic part
  is used to control the packet, the higher order part fulfilling very good estimates.
\end{abstract}

\noindent

\medskip

 
\section{Introduction}

The numerical experiment performed by Fermi, Pasta, Ulam and  Tsingou was originally designed to investigate the dynamical
foundations of equilibrium statistical mechanics by studying the
exchange of energy among the normal modes of a weakly nonlinear
particle chain~\cite{FPU55}. The authors considered initial data concentrated
in a packet of low frequency modes, and their data displayed the persistence
of such a concentration within the available computation time. This
is what we call the FPUT phenomenon. Since then, the problem of
understanding the dependence of the latter phenomenon on the number
$N$ of particles, on the total energy $E$ and on the choice of the initial conditions has been the object of a huge
number of both numerical and analytical studies (see e.g.:
\cite{GGMV92,BI05,BCMM15,BCM26,BLP09,BP11,PCSF11,BCP13,BPP18}), but
rigorous results are still quite unsatisfactory. Indeed it has only
been proved that that the FPUT packet persists when the energy $E$
satisfies the condition  $EN^3\ll 1$ \cite{BP06,BM16} {(see also \cite{SW00,GPR21,GP22}),} thus in a regime where $E$ goes to zero as
$N\to\infty$. In the present paper we prove a result ensuring the
persistence of a phenomenon of FPUT type in the regime $E\ll1$,
uniformly with respect to $N$ and for times slightly shorter than
$E^{-1}$. We remark that our result does not survive the thermodynamic
limit ($E/N$ finite, small as $N\to\infty$).

Following
\cite{BCP13,HK08a,HK08b,HK08c,BKP09,BKP15a,BKP15b,BM16,GMMP20}, we take
advantage of the closeness of the FPUT chain to the Toda
lattice. Before introducing our method, we recall the way such a
closeness was exploited in previous works. In particular, in \cite{BM16} a
result on the persistence of the FPUT packet was proved. The idea
there was to construct Birkhoff variables for the Toda lattice, a nonlinear version of the
Fourier modes, and to exploit the fact that the
energy of each Birkhoff variable stays approximately
constant along the FPUT flow for times of the order $(E/N)^{-1}$. The
main limitation of such a result is due to the fact that the Birkhoff variables of the Toda lattice develop a singularity at a distance $\cO(N^{-3/2})$ from the origin in
the energy norm, so that this method only applies to states with
energy going to zero as $N^{-3}$.

In the present paper we overcome such a difficulty avoiding to
construct Birkhoff variables. Instead, we use the Flaschka integrals
$J^{(m)}$ of the Toda lattice \cite{Fla74} in order to control the energy residing in the
Fourier modes. Actually, the idea is to construct a suitable linear combination
of the Flaschka integrals which is useful to control the energy of
the FPUT packet of modes. We consider only the Flaschka integrals $J^{(m)}$ of
even order $m$ which, according to the results of \cite{GMMP20}, have a
quadratic part with a good structure. Then, for any integer $s$, we look for a
linear combination of the first $2s+2$ integrals $J^{(m)}$ displaying a quadratic part
which is essentially a Fourier multiplier by $\sin^{2s}(\theta_k)$,
$\theta_k$ being the normalized wave number of the $k$-th Fourier mode. From a technical point of view, in order to be able to estimate the
drift of such integrals along the FPUT dynamics, we need a precise
estimate of their nonlinear part, and this in turn requires to know the explicit form of the coefficients of the linear combination.

We now briefly describe the proof. First we
remark that, since the Flaschka integrals are given by 
\begin{equation}
 J^{(m)}=\frac1m\Tr (L^m)\ ,
 \label{J2m}
\end{equation}
where $L$ is the Lax matrix of the Toda lattice, a linear combination
of the Flaschka integrals has the form
\begin{equation}
  \label{is}
\cI^{(s)}=\Tr(P_s(L))\ ,
\end{equation}
where $P_s$ is a certain polynomial of degree $2s+2$ in its argument. 
The computation of the quadratic part $\cI_2^{(s)}$ of \eqref{is}, with the
explicit determination of the coefficients of the desired linear combination, and the estimate of the relative remainder, constitute the main technical part of the present work. 
More precisely, we expand $P_s$ on
the basis of Chebichev polynomials and use their properties in order to
determine the coefficients of the expansions and to get a quadratic
part with the desired property described above. It turns out that the
polynomials $P_s$ thus obtained are essentially Gauss hyper geometric functions.
The estimate of the remainder of the expansion is then obtained by
exploiting the main result of \cite{PSS13} and some explicit
integral representations of the hypergeometric functions. 

\vskip10pt

We conclude this introduction by first recalling the papers
\cite{BG93} dealing with the problem of energy sharing in a
  diatomic chain.  
We also mention two further lines of research on
the FPUT problem with initial conditions "in measure". The first one is mainly developed in the series of papers \cite{Car07,CM12,MBC14,Mai19,GPP15,GPP26}, where the dynamics of a chain
of particles was studied in the thermodynamic limit. In such papers
the authors consider the system in thermal equilibrium,
namely with data extracted according to the Gibbs measure, and get
an upper bound on the probability that the fluctuations of the energy
of a packet of modes become large. On the other hand, these results do
not provide any information on the persistence of the FPUT
packet.  The second line consists in using the ideas of wave
turbulence in order to understand the process of destruction of the
wave packet \cite{OVPL15,PBCLO19,OLDC23,Wu25,WOHP25}. Here the
heuristic techniques are very well developed, but the rigorous results
have not yet reached a conclusive stage \cite{Wu25,WOHP25,GLM26}.

\vskip10pt

The paper is organized as follows. In Section~2 we introduce the
periodic FPUT chain, define the energy of a packet of modes and state
the main theorem. In Section~3 we formulate the construction and the
estimates for the recombined Toda integrals. In Section~4
the coefficients of $P_s$ are determined through Chebyshev and Legendre polynomial
identities. In Section~5 we estimate the nonlinear remainder, using
the hypergeometric representation and spectral-shift theory. In Section~6
the FPUT and Toda flows are compared, and the proof of the main
theorem is concluded. The appendices collect a divided-difference formula for
traces of matrix polynomials, a formula for the differential of a
Taylor remainder, the spectral-shift estimate used in the proof, and
the proof of
some technical lemmas.

\medskip

\noindent{\it Declaration of generative AI use.}  During the
development of this work, the authors used OpenAI’s ChatGPT to explore
and refine some intermediate arguments. All mathematical statements
and proofs were independently verified by the authors, who take full
responsibility for the content.

\medskip

\noindent{\it Acknowledgments.} We warmly thank Livio Pizzochero for
several discussions and for suggestions which allowed to improve the
presentation and the statement of the main result. The authors were
partially supported by INdAM-GNFM.

\medskip
\noindent{\it Data Availability statement.} We have no data associated
to our paper.

\section{Main result}\label{main.2}


We consider a chain of $N$ particles with periodic boundary
conditions. Denoting by $(p_j,q_j)_{j=1,...,N}$ their momentum and 
displacement respectively, the
Hamiltonian of the FPUT $\alpha$-model is given, with a suitable
normalization, by 
\begin{equation}
  \label{FPU}
H(p,q)=\sum_{j=1}^{N} \frac{p_j^2+(q_{j+1}-q_j)^2}{2}-
\sum_{j=1}^{N}  \frac{(q_{j+1}-q_j)^3}{6}\ ,
\end{equation}
where $q_{j+N}=q_j$ and $p_{j+N}=p_j$. Denoting by
\begin{equation}
  \label{tetak}
\theta_k:=\frac{\pi k}{N}\ \ (k=0,1,\dots,N-1)
  \end{equation}
the normalized wave number, we define the Fourier variables $\hat q_k,\hat p_k$  by
\begin{align}
  &\hat q_k=\frac{1}{\sqrt{N}}\sum_{j=0}^{N-1}q_je^{-2i j\theta_k}\ \ ;\ \   
  \hat p_k=\frac{1}{\sqrt{N}}\sum_{j=0}^{N-1}p_je^{2i j\theta_k}\ ; \label{fou.1k}
  \\
  &q_j=\frac{1}{\sqrt{N}}\sum_{k=0}^{N-1}\hat q_k e^{2i j\theta_k}\ \ ;\ \ 
  p_j=\frac{1}{\sqrt{N}}\sum_{k=0}^{N-1}\hat p_k e^{-2i j\theta_k}\ , \label{fou.1j}
\end{align}
where $i$ is the imaginary unit. Thus
\begin{align}
 \label{H2}
H_2(p,q):= \sum_{j=0}^{N-1}
\frac{p_j^2+(q_{j+1}-q_j)^2}{2}=\sum_{k=0}^{N-1} E_k\ ;
\\
  \label{Ek}
  E_k:=\frac{|\hat p_k|^2+\omega_k^2|\hat q_k|^2}{2}\ ,\quad
  \omega_k:=2\sin\theta_k\ .
\end{align}

\noindent
The total momentum $P:=\sum_{j=1}^Np_j$ is an integral of motion, so
we can restrict to the surface $P=\sqrt{N}\hat p_0=0$; then, also the surface
$Q=\sum_{j=1}^Nq_j=\sqrt{N}\hat q_0=0$ is invariant and we restrict to such a surface.
Thus, the phase space of the system is the vector subspace of $\R^{N}\times \R^N$
composed by the sequences $(p_j,q_j)_{j=1,...,N}$ fulfilling the
conditions $P=0=Q$. The map that associates to such a sequence its
Fourier coefficients $(\hat p_k,\hat q_k )_{k=1,...,N-1}$ is a
global coordinate system. 

\begin{remark}
{ The potential energy of the Hamiltonian \eqref{FPU} is unbounded
  from below, but displays a local minimum at the origin. As a
  consequence, if $E$ is small enough, the constant energy surface
  $\{H=E\}$ contains an invariant compact connected component
  $\Sigma_E$ surrounding the origin. Thus in particular each solution of the Hamilton equations starting on it is global. One can easily show that the condition on $E$ ensuring the existence of $\Sigma_E$ is $0<E<2/3$.} 
\end{remark}

In order to state our main result we have to define the energy of a
packet of modes. We preliminarily observe that, due to the symmetries
of the Fourier transform \eqref{fou.1k}-\eqref{fou.1j} corresponding
to real initial data, one has $E_k=E_{N-k}$ for any $k=1,\dots,N-1$. Then, we fix a positive $\delta\ll1$ ruling the width of the packet, and, for any $0<\theta<\frac{\pi}{2}-\delta$ we define the \emph{packet} of modes of wave number $\delta$-close to $\theta$ as
\begin{align}
  \label{kappone}
\cK_{\theta}:=\left\{ k\in \{1,...,N-1\}\ :\ \theta_k\in
   [\theta,\theta+\delta]\cup[\pi-\theta-\delta,\pi-\theta]\right\}\ .
\end{align}
Then, the energy of the packet $\cK_\theta$ is given by
\begin{equation}
  \label{Etheta}
\cE_{\theta}:=\sum_{k\in\cK_\theta}E_k\ .
\end{equation}
Our main theorem is the following.
\begin{theorem}
  \label{main.1}
For any $0<\delta\ll1$ there exists a positive constant $C$ with the
following property. For any $s\in\N$, with $2s<N$ and for any $\theta_*\in\left(0,\frac\pi
2 -\delta\right)$, and any $\epsilon\in(0,1)$, if the initial
datum fulfills
\begin{align}
  \label{ini}
&\hat p_k(0)=\hat q_k(0)=0\quad
\text{if}\quad \theta_*<\theta_k<\pi-\theta_* \ ;
\\
  \label{ini.2}
&\cE:=H_2(p(0),q(0))<\frac{1}{C}\min\left\{\sin^{4s/\epsilon}\theta_*;s^{-\frac{2}{1-\epsilon}}\right\}\ ,
\end{align}
then,
\begin{align}
  \label{ss}
\forall \theta\in\left(\theta_*,\frac{\pi}{2}-\delta\right)
\end{align}
one has
\begin{align}
  \label{great}
\frac{\cE_\theta(t)}{\cE}\leq
C\left(\frac{\sin\theta_*}{\sin\theta}\right)^{2s} \ ,\quad |t|\leq
\frac{1}{C\cE^{1-\epsilon/2} }\ .
\end{align}
\end{theorem}
\begin{remark}
{ One easily shows that the relation between $\cE$ - the initial value of the quadratic energy - and the total energy $E$ is $\cE=E+O(E^{3/2})$.} 
\end{remark}
\begin{remark}
  \label{decay}
The above theorem shows in particular that the energy in any packet of
modes with wave number $\theta>\theta_*$ decays as a large
power of the ratio $sin\theta_*/\sin\theta$, showing that the energy
resides in the modes with wave number smaller
than $\theta_*$, up to an error of the order $\cE^{\epsilon/2}$. 
\end{remark}
\begin{remark}
  \label{uniformi}
We stress that all the constants appearing in Theorem \ref{main.1} are uniform with respect to 
$N$.
\end{remark}
\begin{remark}
  \label{tempi}
The times of validity of the estimate \eqref{great} are the typical
ones covered by averaging theorems, but we actually do not know if the
result can be obtained by the methods of perturbation theory or if the
actual time of persistence of the FPUT packet is longer than what we
ensure.  
\end{remark}
\begin{remark}
  \label{thermo}
Our proof does not survive the thermodynamic limit. {Apart from the fact that the Hamiltonian \eqref{FPU} is unbounded from below when $E$ is large, even adding a quartic stabilizing term to $H$, in order to get 
results in this limit we guess that the introduction of some
probabilistic tools might be necessary}. We do not have clear ideas on how to implement
such a procedure at present. 
\end{remark}

\section{A linear combination of Toda integrals}

We recall that the Toda lattice with periodic boundary conditions
and $N$ particles 
is the Hamiltonian system with Hamiltonian
\begin{equation}
  \label{Toda}
H_{Toda}=\frac 1 2\sum_{j=1}^Np_j^2+\sum_{j=1}^Ne^{-r_j},
  \end{equation}
where
$$
r_j:=q_{j+1}-q_j\ .
$$
The Hamilton equations
of motion associated to $H_{Toda}$ admit a Lax pair formulation which,
when expressed in the Flaschka coordinates, $ b_j=-p_j\mbox{ and }
a_j=e^{-\frac 12r_j} $, takes the form \cite{Fla74} with
\begin{equation}\label{laxx}
\dot L = [B,L]
\end{equation}
 where the $N \times N$ matrix $L $ is given by
 \begin{equation}
   \label{L}
L=L(a,b):= \begin{pmatrix} 
   b_1 &a_1 &0&\ldots&0&a_N \\ 
   a_1 &b_2 &a_2&\ldots &0&0 \\ 
   0 &a_2 &b_3&\ldots&0&0\\
   \vdots&&&&&\vdots\\
   0&\ldots &&\ldots &b_{N-1}&a_{N-1}\\
      a_N &0&\vdots &&a_{N-1}&b_N 
      \end{pmatrix}\ ,
   \end{equation}
   while the form of the skewsymmetric matrix $B$ plays no role in
   our developments.  In the following we will consider $L=L(p,q)$ as
a function of the variables $p,q$ through the corresponding Flaschka
variables. In particular all the eigenvalues of the Lax matrix are
invariant under the flow of the Toda Hamiltonian \eqref{Toda}.

The classical first integrals of the Toda chain, first introduced by Flaschka \cite{Fla74}, are the quantities \eqref{J2m}.

For any $s\geq1$ we want to construct a linear combination
\begin{equation}
  \label{linear.comb}
\cI^{(s)}(p,q):=\sum_{l=0}^sc_{s,l}{J^{(2l+2)}(p,q)}\ ,
  \end{equation}
of the
integrals $J^{(2l)}$ with even order with the property that the
quadratic part of $\cI_2^{(s)}$ is essentially a Fourier multiplier by
$\sin^{2s}\theta_k$.

In order to give a precise statement, in the space of the states with zero average $P=0=Q$, we define the
energy norm by
\begin{equation}
  \label{norma}
\left\|(p,q)\right\|^2:=\sum_{j=0}^{N-1}\frac{p_j^2+r_j^2}{2}=\sum_{k=1}^{N-1}E_k \ .
\end{equation}

The following theorem is our main technical result.
\begin{theorem}
  \label{main}
For any $s\in\N$ with $2s<N$,  there exists a linear combination $\cI^{(s)}$ of the form 
\eqref{linear.comb} with the property that its Taylor expansion at
$p=q=0$ restricted to the hypersurface $P=0=Q$ has the form
\begin{equation}
  \label{espando}
\cI^{(s)}(p,q)=\cI^{(s)}_0+\cI^{(s)}_2(p,q)+\cR^{(s)}(p,q)\ ,
\end{equation}
where $\cI^{(s)}_0$ is independent of $(p,q)$, the quadratic part
$\cI^{(s)}_2$ is given by
\begin{equation}
  \label{parte.lin}
\cI^{(s)}_2(p,q)={\sum_{k=1}^{N-1} \sin^{2s}\theta_k\, E_k}
  \end{equation}
and there exist a positive $C$ and a positive $R_*$ independent of
$s$ and $N$ s.t. for $R< R_*$ one has 
\begin{align}
  \label{sti.resto}
\sup_{\|(p,q)\| \leq R}\left|\cR^{(s)}(p,q)\right|&< s\, C\,  R^{3}\ .
\\
\label{sti.resto.dr}
\sup_{\|(p,q)\| \leq R}\left|d\cR^{(s)}(p,q)[Y]\right|&< s\, C\,
R^2\left\| Y\right\|\ .
\end{align}
The coefficients $c_{s,l}$ are explicitly given by
\begin{equation}
  \label{coeff}
c_{s,l}=\frac{(-1)^l\, (2s+1)!\, l!}{4^s s!(s-l)!(2l+1)!}\ .
  \end{equation}
\end{theorem}

\begin{remark}
  \label{no1}
The fact that the part of $J^{(2m)}$ linear in $p,q$ vanishes on the
surface $P=0=Q$ was proved in \cite{GMMP20}.
\end{remark}

\begin{remark}
  \label{noi_no}
In the case of the original integrals $J^{(2m)}$ (cf. \eqref{J2m}) our
method would produce an estimate of the remainder of the form
\eqref{sti.resto} and \eqref{sti.resto.dr} whose r.h.s. would contain
$2^{2m} C$ in place of $sC$. This is due to the fact that the spectrum
of $L(0,0)$ asymptotically fills $[-2,2]$.  This estimate would be
useless to our purposes.
\end{remark}

\section{Computation of the coefficients $c_{s,l}$.}\label{icoeff} 

First we recall that the structure of the integrals \eqref{J2m} and their expansion at
the equilibrium was studied in detail in
\cite{GMMP20}, where, in particular  the following theorem was proved.

\begin{theorem}
  \label{GMP}[Lemma 3.7 of \cite{GMMP20}]
Let us consider the Toda integral $J^{(2m)}$ and its Taylor expansion
at $p =q = 0$: it takes the form
\begin{equation}
\label{Jm.exp}
J^{(2m)}(p, q) = J^{(2m)}_0 + J^{(2m)}_1(p,q)+J^{(2m)}_2(p,q)+  J^{(2m)}_{\geq 3}(p,q),
\end{equation}
where 
$ J^{(2m)}_0=
c \in \R $, $J^{(2m)}_1(p,q)$ is linear in $(p,q)$ and vanishes on the
surface $P=Q=0$.
$J^{(2m)}_2(p,q)$ is quadratic in $(p,q)$ and has the form 
\begin{equation}
\label{Jm2.struct}
J^{(2m)}_2(p,q) = 
  \sum_{i,j}p_i  A^{(2m)}_{i,j} p_j + \sum_{i,j} r_i  A^{(2m)}_{i.j} r_j 
\end{equation}
with $A^{(2m)}$ a  symmetric  $N \times N$  matrix.
\end{theorem}

The first idea on how to construct the desired linear combination is that, due to
equation \eqref{J2m} and to the linearity of the trace operator,
the integral \eqref{linear.comb} that we are looking for is given by
\begin{equation}
  \label{traccia}
  \cI^{(s)}=\Tr\left(\sum_{l=0}^sc_{s,l}\frac{L^{2l+2}}{2l+2}\right)
  =\Tr\left(P_s(L)\right)\ ,
\end{equation}
where $$P_s(x)=\sum_{l=0}^sc_{s,l}\frac{x^{2l+2}}{2l+2}$$ is a polynomial to be determined.

Denote by $L_0:=L(0,0)$ the Lax matrix at the equilibrium,
namely 
   \[ L_0 := \begin{pmatrix} 
   0 &1 &0&\ldots&0&1 \\ 
   1 &0 &1&\ldots &0&0 \\ 
   0 &1 &0&\ldots&0&0\\
   \vdots&&&&&\vdots\\
   0&\ldots &&\ldots &0&1\\
      1 &0&\vdots &&1&0 
      \end{pmatrix}\ .
       \]      

Given an arbitrary polynomial $P$, consider the quantity
$ \cI:=\Tr(P(L))$; we want to compute its quadratic part. To this end
we have to compute a formula for the
quadratic part of the expansion of $\Tr(P(L))$ at $L_0$. Before looking
for such a formula, we observe that, due to Theorem
\ref{GMP}, $J^{(2m)}$ does not have a linear part, while $J^{(2m)}_2$
is a quadratic form in $p$ plus an {\it identical} quadratic form in
$r$. Therefore, it is enough to compute the quadratic form in $p$.
We compute such a quadratic form working with the basis of the eigenvectors
of $L_0$. In order to diagonalize $L_0$, we remark that $L_0=S+S^*$
where $S$ is the shift operator defined by $(Sq)_j:=q_{j+1}$ and
$S^*=S^{-1}$ its adjoint, so that we have the following elementary lemma
\begin{lemma}
  \label{L0}
 The $N$ (complex, orthonormal) eigenvectors ${\bf e}_k$ and eigenvalues $\lambda_k$ of $L_0$ 
  ($k=0,\dots,N-1$) are given by
  \begin{align}
&({\bf e}_k)_j:=\frac{1}{\sqrt{N}}e^{-i j2\theta_k}\ ;\ \ (j=0,\dots,N-1) \label{eigenvec}\\
&\lambda_k=2\cos(2\theta_k) \label{eigenval}\ .
    \end{align}
\end{lemma}

\noindent
The following lemma provides a formula for the term of $\cI$ quadratic in $p$, in the case of a generic polynomial. 

\begin{lemma}
  \label{parte.quad}
Let $P$ be a polynomial and let $\cI:=\Tr(P(L))$, then the part of
$\cI$ quadratic in $p$ is given by
\begin{equation}
  \label{parte.2.p}
{\frac{1}{2}}\sum_{k=0}^{N-1}\left(\frac{1}{N}\sum_{k_1=0}^{N-1}
\frac{P'(\lambda_{k_1+k})-P'(\lambda_{k_1})}{\lambda_{k_1+k}-\lambda_{k_1}}\right)
{ \left|\hat p_k\right|^2}\ ,
\end{equation}
with the convention that when $\lambda_{k_1}=\lambda_{k_1+k}$ the fraction
is $P''(\lambda_{k_1})$.
\end{lemma}
\proof We apply Lemma \ref{divided.1} of the Appendix \ref{divided}
with $A=L_0$ and {$X=-\text{diag}(p_1,...,p_N)$. Here $X$ represent
  the part of $L-L_0$ which depends only on $p$: in general one would
  have $L-L_0=X(p)+Y(q)$, but we focus only on $X$ for the reason
  explained above. The matrix elements of $X_{k_1,k_2}$ in the basis
  \eqref{eigenvec} are
\begin{align}
  \label{espandoX}
X_{k_1,k_2}=&\langle {\bf e}_{k_1},X{\bf e}_{k_2}\rangle=\sum_{j,l=0}^{N-1}
\overline{({\bf e}_{k_1})_j}(-\delta_{j,l}p_l)({\bf e}_{k_2})_l= \\
=&-\frac{1}{N}\sum_{j=0}^{N-1}p_je^{i j2(\theta_{k_1}-\theta_{k_2})}=
-\frac{1}{\sqrt{N}}\hat p_{k_1-k_2}\ .
\end{align}
Since the $p_j$'s are real, it follows that 
$$
X_{k,k'}X_{k',k}=\frac{1}{N}\hat p_{k-k'}\hat p_{k'-k}=
\frac{1}{N}\left|\hat p_{k-k'}\right|^2\ ,
$$
and equation \eqref{divided.2} immediately gives the result \eqref{parte.2.p}.} \qed

Now, in order to get an explicit expression for the bracket in \eqref{parte.2.p}, it is convenient to expand $P$ on the basis of the Chebyshev polynomials of
the first kind, which are defined as the polynomials $T_n(x) $ such
that (see \cite{AS64} Eq. 22.3.15)
\begin{equation}
  \label{tche}
T_n(\cos\alpha)=\cos(n\alpha)\ .
\end{equation}
More precisely, writing
\begin{equation}
  \label{Pj}
P(x)=\sum_{j=0}^s\gamma_{j} B_j(x)\ ,
\end{equation}
with $B_j$ defined by
\begin{equation}
  \label{Bj}
B'_{j}(x)=2T_{2j+1}\left(\frac{x}{2}\right) \ ,
\end{equation}
we compute the contribution of \eqref{Bj} to \eqref{parte.2.p}.
\begin{lemma}
  \label{cotri.Bj}Assume $2s<N$, then
  the quadratic part of $\Tr(B_j(L))$ is
  \begin{equation}
    \label{qua.bj}
{\sum_{k=0}^{N-1}W_j(\theta_k)E_k}\ ,
  \end{equation}
  where
  \begin{equation}
    \label{wj}
    W_j(\theta)=\frac{\sin(2j+1)\theta}{\sin\theta}
    =1+2\sum_{n=1}^j\cos(2n\theta)\ .
  \end{equation}
Therefore, the quadratic part of $\Tr(P(L))$ with $P=\sum_{j=0}^s
  \gamma_jB_j $ is given by
  \begin{equation}
    \label{qua.bj.1}
{\sum_{k=0}^{N-1}\left[\sum_{j=0}^s\gamma_jW_j(\theta_k)\right] E_k}.
    \end{equation}
\end{lemma}
\proof We preliminarily remark that the second equality in \eqref{wj} is
nothing but the identity yielding the Dirichlet kernel (see e.g. \cite{Str92}, p.137).

We come to the proof of \eqref{qua.bj}.
We have to compute the expression \eqref{parte.2.p} with
$P=B_j$ (and $j\leq s$). Let us denote, for simplicity
$$
q=2\theta_{k_1}\ \ ;\ \ \eta:=2\theta_{k}\ ,
$$
so that, by \eqref{tche}
$$
B'_{j}(\lambda_{k_1})=\cos((2j+1)q)\ \ ;\ \ 
B'_{j}(\lambda_{k_1+k})=\cos((2j+1)(q+\eta))\ . 
$$
Then the argument of the sum in bracket of \eqref{parte.2.p} is given by
\begin{align*}
\frac{T_{2j+1}(\cos(q+\eta))-T_{2j+1}(\cos(q))}{\cos(q+\eta)-\cos q}=
\frac{\cos((2j+1)(q+\eta))-\cos((2j+1)q)}{\cos(q+\eta)-\cos
  q}
\\
=\frac{\sin((2j+1)
  \eta/2)}{\sin(\eta/2)}\frac{\sin((2j+1)(q+\eta/2))}{\sin(q+\eta/2)}
=W_j\left(\frac{\eta}{2}\right)W_j\left(q+\frac{\eta}{2}\right)\ .
\end{align*}
Now, the first factor, i.e. $W_j(\eta/2)$, is independent of $k_1$. We claim that the sum
over $k_1$ of the second factor is equal to $N$. This is an immediate
consequence of the second equality in \eqref{wj}. Indeed, if $2n$ is
not a multiple of $N$, the
contribution of each term is
$$
\sum_{k_1=0}^{N-1}\cos\left(2n (2\theta_{k_1}+\theta_k)\right)=0\ .
$$ The condition on $n$ is automatically satisfied since $n\leq j\leq s<N/2$.  It follows that only the constant term of \eqref{wj}
gives a nonzero contribution amounting to $N$, which is in turn canceled by the
$\frac{1}{N}$ factor inside the bracket in \eqref{parte.2.p}. \qed

In view of this lemma we now want to determine coefficients
$\gamma_{s,j}$
s.t.
\begin{equation}
  \label{gamma}
\sum_{j=0}^s\gamma_{s,j}W_j(\theta)=\sin^{2s}(\theta)\ .
  \end{equation}
\begin{lemma}
  \label{gamma.1}
  The coefficients such that \eqref{gamma} holds true are explicitly given by
  \begin{equation}
    \label{gamma.2}
\gamma_{s,j}:=\frac{(-1)^j}{4^s}\begin{pmatrix}
  2s+1\\ s-j\end{pmatrix} \ .
    \end{equation}
\end{lemma}
\proof Observe that, in view of the definition \eqref{wj} of
$W_j$, \eqref{gamma} is equivalent to 
$$
\sum_{j=0}^s\gamma_{s,j}\sin((2j+1)\theta)=\sin^{2s+1}(\theta)\ .
$$
Starting from Euler formula, by a lengthy computation, one gets
\begin{equation}
  \label{fine.seni}
\sin^{2s+1}\theta
=
\frac{1}{2^{2s}}
\sum_{j=0}^{s}
(-1)^j
\binom{2s+1}{s-j}
\sin\bigl((2j+1)\theta\bigr)\ ,
  \end{equation}
which implies  the thesis.
\qed
\begin{corollary}
  \label{prossima}
Let  $P_s$ be the polynomial s.t.
\begin{equation}
  \label{cor.W}
P_s'(x)=\sum_{j=1}^s \frac{(-1)^j}{4^s}\begin{pmatrix}
  2s+1\\ s-j\end{pmatrix} 2T_{2j+1}\left(\frac{x}{2}\right)\ ,
\end{equation}
then the quadratic part of $\Tr(P_s(L))$ is given by
\eqref{parte.lin}.
\end{corollary}
Finally, we need to expand the Chebyshev polynomials in order to get
$P_s$ on the standard basis of the monomials $x^l$. To such an end, we need
the formula 22.3.6 of \cite{AS64} which, in our case, reads
\begin{equation}
  \label{tjAS}
2T_{2j+1}\left(\frac{x}{2}\right)=(2j+1)\sum_{l=0}^j(-1)^{j-l}
\frac{(j+l)!}{(j-l)!(2l+1)!}x^{2l+1}\ .
  \end{equation}

By Corollary \ref{prossima} one gets
\begin{equation}
  \label{poli.d}
P'_s(x)=\sum_{l=0}^s\frac{(-1)^l}{4^s}\frac{x^{2l+1}}{(2l+1)!}\sum_{j=l}^s\binom{2s+1}{s-j}(2j+1)\frac{(j+l)!}{(j-l)!}\ ,
\end{equation}
so the last step is to simplify the sum over $j$ in the previous
expression. This is done in Lemma \ref{Ssl} below (in the Appendix \ref{E}).

\noindent{\it End of the proof of Eq. \eqref{coeff}.} The conclusion is obtained by inserting the result of Lemma \ref{Ssl} in Eq. \eqref{poli.d}. \qed

\section{Estimate of the nonlinear part}\label{nonlinear}

The key tool to prove the estimate \eqref{sti.resto} is Theorem 1.1 of
\cite{PSS13}, which is recalled in the Appendix \ref{B}. We denote
$T(L):=\Tr(P_s(L))$. In the appendix we
also prove, in Corollary \ref{cor.PSS}, the estimate
\begin{align}
\left|\cR_{T,n}(L_0,X)\right|=\left|\Tr \left( P_s(L_0+X)-\left.\sum_{k=0}^{n-1}\frac{1}{k!}
\frac{d^k}{dt^k}\bigl[P_s(L_0+tX)\bigr]\right|_{t=0}  \right)\right|
\\
\label{gauss.7}
\leq
C_{\delta,n}
\left(\sup_{k=0,...,n}\,\sup_{x\in[-2-\delta,2+\delta]}\left|P^{(k)}_s(x)\right|
\right)
\left\|X\right\|_{S^n}^n\ .
\end{align}
Here $\delta$ is a small parameter,{ $\|\cdot\|_{S^n}$
  denotes the Shatten-Von Neuman norm (see Eq. \eqref{B.1}) which,
  for $s\geq 2$ is controlled by the Hilbert-Schmidt norm} and we will use this formula for
$n=1,2,3$.

We now write a formula for $\cR^{(3)}$. To this end denote
\begin{align}
  \label{inutile}
  X_1:=L(\frac{-r_j}{2},-p_j)\ ,\quad X_2:=L(\frac{r_j^2}{8},0)\ ,
  \\
  X_3:=L(a_j-\left(1-\frac{r_j}{2}+ \frac{r_j^2}{8} \right),0),
\end{align}
so that $X:=L-L_0=X_1+X_2+X_3$. We have

\begin{align}
  \label{espando.1}
  \cI^{(s)}= T(L_0)+dT(L_0)X_1+dT(L_0)X_2+dT(L_0)X_3
  \\
  +\frac{1}{2} d^2T(L_0)(X_1,X_1)+\sum_{k+l\geq 3}
\frac{1}{2} d^2T(L_0)(X_j,X_k)+\cR_{T,3}(L_0,x).
\end{align}
Now, $dT(L_0)X_1$ vanishes due to Theorem \ref{GMP}, while
$$\cI^{(s)}_2=dT(L_0)X_2+\frac{1}{2} d^2T(L_0)(X_1,X_1)$$. The other
non constant terms are the different components of $\cR^{(s)}$.

To estimate the single terms we make use of the following lemma

\begin{lemma}
  \label{HS}
One has
\begin{align}
 \label{HS.1}
 \left\| X_n\right\|_{HS}&\leq C_n\left\|(p,q)\right\|^n\ ,\quad n=1,2
 \\
 \label{HS.2}
 \left\| X_3\right\|_{S^1}&\leq C_3\left\|(p,q)\right\|^n.
\end{align}
\end{lemma}
\proof We start by $X_1$: denoting by $X_{1,ij}$ its matrix elements, one has
$$
\left\| X_1\right\|_{HS}^2=\sum_{ij}|X_{1,ij}|^2=\sum_{j=1}^Np_j^2+2\sum_{j=1}^N\frac{r_j^2}{4}\ ,
$$
from which the result immediately follows. $X_2$ is identically
estimated. Concerning $X_3$, we just remark that
$\left|a_j-(1-r_j/2+r_j^2/8)\right|\leq C|r_j|^3$, so that the result
immediately follows. 
\qed

Using this Lemma and \eqref{gauss.7} we immidiately get.
\begin{equation}
  \label{vero.resto}
\left|\cR^{(s)}(p,q)\right|\leq C \left(\sup_{k=0,...,3}\,\sup_{x\in[-2-\delta,2+\delta]}\left|P^{(k)}_s(x)\right|
\right)
\left\|(p,q)\right\|^3\ ,
 \end{equation}
and a similar estimate for its differential.

Then, the key remark to get a control of the supremum of the polynomial, is that
the polynomial $P_s$, whose estimate has to be inserted in
\eqref{gauss.7}, admits a representation in terms of the Gauss
hypergeometric series. First we recall that the Gauss hypergeometric series
$\Fdue$ is defined by (see \cite{AS64} Eq. 15.1.1)
\begin{align}
  \label{gauss}
\Fdue(a,b;c;z):=\sum_{n=0}^{\infty}\frac{{(a)_n}(b)_n}{(c)_n}\frac{z^n}{n!}
\\
\label{gauss.1}
(a)_n:=a(a+1)...(a+n-1)\ ,\quad \forall n\geq 1
\end{align}
and $(a)_0:=1$. In particular, if $a=-s$ or $b=-s$ for some integer $s$
this is a polynomial. The function $\Fdue$ fulfills (see \cite{AS64} eq. 15.3.1)
\begin{equation}
  \label{gauss.2}
\Fdue(a,b;c;z)=\frac{\Gamma(c)}{\Gamma(b)\Gamma(c-b)}\int_0^1\frac{t^{b-1}
  (1-t)^{c-b
-1}}{(1-tz)^a}dt\ ,\quad \Re(c)>\Re (b)>0\ .
\end{equation}

\begin{lemma}
  \label{gauss.3}
  Define the polynomial $Q_s(z)$ by $P'_s(x)=xQ_s(x^2)$, then one has
  \begin{equation}
    \label{gauss.4}
Q_s(z)=A_{s}\Fdue\left(-s,1;\frac{3}{2};\frac{z}{4}\right)\ ,
  \end{equation}
  with
  $$
A_s:=\frac{2s+1}{4^s}\binom{2s}s\ .
$$
Furthermore, one has
\begin{equation}
  \label{gauss.6}
\Fdue\left(-s,1;\frac{3}{2};z\right)=\frac{1}{2}\int_0^1\frac{(1-tz)^s}
      {\sqrt{1-t}}dt\ . 
  \end{equation}
\end{lemma}
\proof In order to get \eqref{gauss.4}, one has to compute the various terms in the definition\eqref{gauss} of $\Fdue$ and to insert them in the definition of $Q_s$
using the explicit formula \eqref{coeff} for the coefficients. The
only nontrivial part is the computation of
$$
\left(\frac{3}{2}\right)_n=\frac{3\cdot(3+2)\cdot(3+4)
  \cdot...\cdot(3+2n-2)}{2^n}  =\frac{1}{2^n}\frac{(2n+1)!}{2^nn!}\ .
$$
Finally, \eqref{gauss.6} is just the specification of \eqref{gauss.2}
to our case. 
\qed

Finally we have to estimate the polynomial $P_s$ and its derivatives.
To this end we first remark that, by Stirling formula, one has
\begin{equation}
  \label{stirlieng}
A_s\leq 2\sqrt s\ .
\end{equation}

To proceed, we remark that, since for $R$ small enough the spectrum of $L$
is arbitrarily close to $[-2,2]$, it is enough to work in a segment
containing $[-2,2]$ (see Eq. \eqref{gauss.7}). Let us define
\begin{equation}
 f(z):=\int_0^1 \frac{(1-tz)^s}{\sqrt{1-t}}\,dt,\quad  F(y):=y f(y^2)
 \qquad s\ge 3,
\label{eq:def-f}
\end{equation}
so that $P'_s(2y)=A_s F(y)$. We are interested in estimates for $F$, $F'$ and $F''$ on an interval of the form
\begin{equation*}
-1 -\delta < y < 1+\delta,
\end{equation*}
where $\delta>0$ is small. By the (skew-)symmetry of $F$ it is enough
to get the needed estimates in {$(-\delta,1+\delta)$}.

\begin{lemma}\label{prop:main-estimates}
Let $0<\delta<\sqrt{2}-1$. There exists a constant
$C_\delta>0$, independent of $s$ and $y$, such that for every $s\ge 3$ and every
$y\in(-\delta,1+\delta)$,
\begin{equation}
 |F(y)|\le \frac{C_\delta}{\sqrt{s}}\ ;
\label{eq:F-estimate}
\end{equation}
\begin{equation}
 |F'(y)|\le C_\delta\ ;
\label{eq:Fprime-estimate}
\end{equation}
\begin{equation}
 |F''(y)|\le C_\delta \sqrt{s}\ ;
\label{eq:Fsecond-estimate}
\end{equation}
Moreover, defining
\begin{equation}
 G(y):=\int_0^y F(u)\,du \ ,
\label{eq:G-explicit}
\end{equation}
one has
\begin{equation}
 |G(y)|\le C_\delta \frac{\log(2+s)}{s}\ .
\label{eq:G-estimate}
\end{equation}
\end{lemma}

\noindent
The quite technical (and essentially standard) proof of Lemma
\ref{prop:main-estimates} is postponed to Appendix \ref{E}.

\noindent{\it End of the proof of Theorem \ref{main}.} Inserting the
result of Lemma \ref{prop:main-estimates} in the estimate of $P_s$ and
in \eqref{gauss.7}, exploiting Lemma \ref{gauss.3} one gets the
estimates \eqref{sti.resto} and \eqref{sti.resto.dr}. This
concludes the proof of Theorem \ref{main}. \qed 

\section{Estimates on the FPUT dynamics}
\label{FPUest}

Let us consider the FPUT and Toda Hamiltonians, defined in \eqref{FPU} and \eqref{Toda}, respectively. Their difference $W$ is given, up to an irrelevant constant, by
$$
W:=H-H_{Toda}=\sum_{j=0}^{N-1}w(r_j)\ ,
$$
where
\begin{equation}
  \label{wn}
w(r_j)=-\sum_{n\geq4}\frac{1}{n!}(-r_j)^n\ .
\end{equation}
Thus, the Hamiltonian vector field $X_W$ of $W$ has only components on
$p$. In particular its $p_j$ component is given by
\begin{equation}
  \label{Wpn}
\left(X_W(r)\right)_{p_j}=-\frac{\partial W}{\partial q_j}=w'(r_j)-w'(r_{j-1})\ .
\end{equation}
It follows that, provided $R$ is small enough,
\begin{align}
  \label{supwn}
\sup_{\left\|(p,q)\right\|\leq R}{\left\|X_{W}(p,q)\right\|}\leq 4
R^{3}\ .
\end{align}
\begin{lemma}
  \label{main.lemma}
There exist positive constants $R_*$ and $C$, s.t. if the initial
datum ${(p(0),q(0))}$ fulfills
\begin{equation}
  \label{R0}
R:=\left\|{(p(0),q(0))}\right\|< R_*\ ,
\end{equation}
then
\begin{equation}
  \label{scelta.s}
\forall\epsilon\in(0,1)\ \ ,\ \ \forall s\in\N
\ s.t.\ s\leq\min\left\{\frac{ N}{2};\frac{1}{C R^{1-\epsilon}}\right\}\ , 
\end{equation}
one has
\begin{equation}
  \label{main.lemma.1}
\cI^{(s)}_2(t)\leq \cI^{(s)}_2(0)+CR^{2+\epsilon}\ ,\quad \forall
\left|t\right|\leq \frac{1}{CR^{2-\epsilon}}\ .
  \end{equation}
\end{lemma}
\proof First we remark that, by conservation of energy and coercivity
of the Hamiltonian, one has that,
provided $R_*$ is small enough
\begin{equation}
  \label{energia}
\left\|(p(t),q(t))\right\|\leq 4 R\ ,\quad \forall t\in\R\ .
\end{equation}
Then, in a ball of radius $4R$ we have
\begin{equation}
  \label{time.der}
\left\{\cI^{(s)},H\right\}=\left\{\cI^{(s)},W\right\}=d\cI^{(s)}X_{W}
=d\cI_2^{(s)}X_{W}+d\cR^{(s)}X_{W}\ ,
  \end{equation}
$\{\ ,\}$ denoting the Poisson bracket.
Now, by Eq. \eqref{parte.lin}, the fact that $\cI^{(s)}_2$ is a
bounded homogeneous polynomial of degree 2, and by Eq. \eqref{supwn}, one
has
\begin{equation}
  \label{poisson}
\sup_{\left\|(p,q)\right\|\leq4R}\left|\cI_2^{s}(p,q)\right|\leq
R^2\ \Longrightarrow
\sup_{\left\|(p,q)\right\|\leq4R}\left|d\cI^{(s)}_2X_W\right|\leq
32 R^4\ .
  \end{equation}
Moreover, by \eqref{sti.resto.dr} one gets
$$
\sup_{\left\|(p,q)\right\|\leq4R}\left|d\cR^{(s)}X_W\right|\leq
CsR^5\ ,
$$
so that, the second of \eqref{scelta.s} being fulfilled, one finally finds
\begin{equation}
  \label{sti.poisson}
\sup_{\left\|(p,q)\right\|\leq4R}\left|\left\{\cI^{(s)};H\right\}\right|\leq
C R^4\ .
  \end{equation}
Now, by means of the triangular inequality, and exploiting the uniformity in time of the bound \eqref{sti.resto}, one gets
\begin{align*}
\cI^{(s)}_2(t)\leq&
\cI^{(s)}_2(0)+\left|\cR^{(s)}(0)\right|+\left|\cR^{(s)}(t)\right|+\left|\cI^{(s)}(t)-
\cI^{(s)}(0)\right|\leq \\
\leq & \cI^{(s)}_2(0)+{2sCR^3}+|t| CR^4\ ,
\end{align*}
from which the thesis immediately follows. \qed

We now specialize the result to the case where the initial datum is
concentrated on the low modes and prove the main result Theorem
\ref{main.1}.

\noindent {\it End of the proof of Theorem \ref{main.1}.}  Fix
$\theta_*<\frac{\pi}{2}-\delta$ and take an initial datum as in
\eqref{ini}, \eqref{ini.2}. Let us define
$$
{\cK_*:=\left\{k\in\{
1,...,N-1\}\ :\ \theta_k\in(0,\theta_*)\cup (\pi-\theta_*,\pi)\right\} }\ .
$$
Then we have
\begin{equation}
  \label{ini.Is}
\cI^{(s)}_2(0)=\sum_{k\in\cK_*}\sin^{2s}(\theta_k)E_k(0)
\leq
\sin^{2s}(\theta_*)
\sum_{k\in\cK_*}E_k(0)=\sin^{2s}(\theta_*)
R^2\ .
  \end{equation}
Thus \eqref{main.lemma.1} gives
$$
\cI^{(s)}_2(t)\leq \sin^{2s}(\theta_*)
R^2+CR^{2+\epsilon}\leq 2R^2\max\{\sin^{2s}(\theta_*),CR^{\epsilon}\}
$$
Therefore, for any $\theta\in(\theta_*,\pi/2-\delta)$, one finds
$$
\cE_\theta(t)=
\sum_{k\in\cK_\theta}E_k(t)=\sum_{k\in\cK_\theta}
\frac{\sin^{2s}(\theta_k)}{\sin^{2s}(\theta_k)} E_k(t)\leq
\frac{1}{\sin^{2s}(\theta)} \cI^{(s)}_2(t)\ ,
$$
from which the thesis immediately follows.
\qed

\section*{Appendices}

\appendix
\section{A divided difference formula}
\label{divided}

Let $A$ be a self-adjoint matrix with eigenbasis $(e_q)$, and let
$\lambda_q$ be the corresponding eigenvalues:
\begin{equation}
 Ae_q=\lambda_qe_q.
\end{equation}

\begin{lemma}
  \label{divided.1}
Let $P$ be a polynomial and let $X$ be self-adjoint with matrix
elements $X_{q,q'}$ on the basis of the eigenvectors of $A$. Define
\begin{equation}
 F(t):=\Tr P(A+tX)\ ,
\end{equation}
then one has
\begin{equation}
  \label{divided.2}
 F''(0)=\sum_{q,q'}
 \frac{P'(\lambda_q)-P'(\lambda_{q'})}{\lambda_q-\lambda_{q'}}X_{qq'}X_{q'q},
\end{equation}
with the convention that if $\lambda_q=\lambda_{q'}$, then
  $$
 \frac{P'(\lambda_q)-P'(\lambda_{q'})}{\lambda_q-\lambda_{q'}}:=P''(\lambda_q)\ .   $$
\end{lemma}
\proof By the linearity and the cyclicity of the trace one has
\begin{equation}
 F'(t)=\Tr\left(P'(A+tX)X\right).
\end{equation}
Thus
\begin{equation}
\label{F2}
  F''(0)=\Tr\left((dP'(A)[X])X\right),
\end{equation}
where $dP'(A)[X]$ is the differential of the map $A\mapsto P'(A)$ at
$A$ with increment $X$. We now compute such a differential. By
linearity it is enough to consider the case where
$P'=G$ with $G(x)=x^n$, then
\begin{equation}
 dG(A)[X]=\sum_{r=0}^{n-1}A^rXA^{n-1-r}.
\end{equation}
In the eigenbasis of $A$,
\begin{equation}
 (dG(A)[X])_{qq'}=\sum_{r=0}^{n-1}\lambda_q^r\lambda_{q'}^{n-1-r}X_{qq'}=
  \frac{\lambda_q^n-\lambda_{q'}^n}{\lambda_q-\lambda_{q'}} X_{qq'}\ ,
\end{equation}
where we used the equality
$$
 \sum_{r=0}^{n-1}\lambda_q^r\lambda_{q'}^{n-1-r}
 =\frac{\lambda_q^n-\lambda_{q'}^n}{\lambda_q-\lambda_{q'}},
 $$
with the convention that the fraction is $n\lambda_q^{n-1}$ if
$\lambda_q=\lambda_{q'}$. Inserting in \eqref{F2} we get the thesis. 
\qed

\section{A formula for the differential of the remainder}\label{C}

In this appendix we prove a formula for the differential of the
remainder of Taylor expansion which is needed in order to prove
\eqref{sti.resto.dr}. Such a formula immediately follows from
the standard Laplace formula for the remainder of a Taylor expansion,
but we were not able to find it in literature.

\begin{lemma}
  \label{C.1}
Let $\cH$ be a real Banach space, let $U\subset \cH$ be open, and let
\[
f:U\longrightarrow \mathbb{R}
\]
be of class $C^k(U)$. Fix $x_0\in U$ and $X\in\cH$ s.t. $x_0+tX\in
U$, $\forall t\in[0,1]$. For $k\geq 2,$ we define the remainder of order $k$ of the Taylor expansion by
\begin{equation}
\mathcal R_k(x_0,X)
:=
f(x_0+X)
-
\sum_{j=0}^{k-1}\frac{1}{j!}\,
 d^jf(x_0)\bigl[\rep{X}{j}\bigr].
\label{eq:remainder-definition}
\end{equation}
Then, for $Y\in\cH$, one has
\begin{equation}
{
 d_X\mathcal R_k(x_0,X)[Y]
=
\frac{1}{(k-2)!}
\int_0^1
(1-t)^{k-2}
 d^kf(x_0+tX)\bigl[\rep{X}{k-1},Y\bigr] \,dt
}.
\label{eq:main-formula}
\end{equation}
\end{lemma}
\proof
For fixed $X,Y\in \cH$, set
\begin{equation*}
g(t):=d_X\mathcal R_k(x_0,tX)[Y].
\end{equation*}
Using \eqref{eq:remainder-definition}, one has
\begin{equation}
g(t)
=
df(x_0+tX)[Y]
-
\sum_{j=1}^{k-1}
\frac{t^{j-1}}{(j-1)!}\,
 d^jf(x_0)\bigl[\rep{X}{j-1},Y\bigr].
\label{eq:g-expanded}
\end{equation}
Since the second term in \eqref{eq:g-expanded} is precisely the Taylor
polynomial of degree $k-2$ at $t=0$ of the first term, one has
\begin{equation}
g(0)=g'(0)=\cdots=g^{(k-2)}(0)=0.
\label{eq:vanishing-jet}
\end{equation}
Moreover, differentiating \eqref{eq:g-expanded} $k-1$ times gives
\begin{equation}
g^{(k-1)}(t)
=
 d^kf(x_0+tX)\bigl[\rep{X}{k-1},Y\bigr].
\label{eq:g-highest-derivative}
\end{equation}
The one-dimensional Taylor formula with integral remainder, together with the vanishing conditions \eqref{eq:vanishing-jet}, yields
\begin{equation}
g(1)
=
\frac{1}{(k-2)!}
\int_0^1
(1-t)^{k-2}g^{(k-1)}(t)\,dt.
\label{eq:taylor-g}
\end{equation}
Substituting \eqref{eq:g-highest-derivative} into \eqref{eq:taylor-g} gives \eqref{eq:main-formula}.
\qed

\section{A Theorem on spectral shift functions}\label{B}

In this appendix we recall a theorem by Potapov, Skripka and Sukochev
\cite{PSS13} (see also \cite{DS09})
which gives the estimate of the remainder of the expansion of the
trace formula that we use in the main text. We slightly simplify the
statement, since we are interested in a finite dimensional situation
(but with unbounded dimension).

We start by recalling the basic notions we need.

Let $\mathcal H$ be a separable Hilbert space, and let $\mathcal B(\mathcal H)$ be the algebra of bounded operators on $\mathcal H$. For a compact operator $A$, let
\begin{equation*}
s_1(A)\ge s_2(A)\ge\cdots\ge 0
\end{equation*}
be the singular values of $A$, namely the eigenvalues of
\begin{equation*}
|A|=(A^*A)^{1/2},
\end{equation*}
counted with multiplicity.

For $1\le p<\infty$, the Schatten--von Neumann space $S^p(\mathcal H)$ is
\begin{equation}
\label{B.1}
  S^p(\mathcal H)
:=
\left\{
 A\in\mathcal B(\mathcal H):\ A\ \text{compact}\ :\ 
\|A\|_{S^p}
:=
\left(
\sum_{j=1}^{\infty}s_j(A)^p
\right)^{\frac{1}{p}}
<\infty
\right\}.
\end{equation}
We remark that the norm can also be defined by 
\begin{equation*}
\|A\|_{S^p}^p
=
\operatorname{Tr}(|A|^p).
\end{equation*}
and that in the case $p=2$ it reduces to the Hilbert--Schmidt norm. 

We also need to define the space of the functions $f\in C^n(\R)$  with
derivatives having bounded
Fourier transform $\hat f$, so we denote 
\begin{equation*}
W_n
:=
\left\{
 f\in C^n(\mathbb R):
 \widehat{f^{(n)}}\in L^1(\mathbb R)
\right\}.
\end{equation*}
which is endowed by the norm
\begin{equation*}
\|f\|_{W_n}
:=
\int_{\mathbb R}
\left|\widehat{f^{(n)}}(s)\right|\,ds.
\end{equation*}

\begin{theorem}\label{th.PSS}[Theorem 1.1 of \cite{PSS13}]
Let $n\in\mathbb{N}$. Let $A$ be a self-adjoint bounded operator and let $X$ be a self-adjoint operator in $S^{n}$. Denote
$$A_t:=A+tX\ , \quad 0\leq t\leq 1\ .$$  Let
$f\in\bigcap_{k=0}^{n}W_k$. Denote
\begin{align}
  \label{}
  F(A+X):=\Tr(f(A+X))
  \\
  \label{C.4}
  \cR_n (A,X):=F(A+X)-\left.\sum_{k=0}^{n-1}\frac{1}{k!}
\frac{d^k}{dt^k}\bigl[F(A_t)\bigr]\right|_{t=0}.
\end{align}
Then there is a unique function
$\eta_n=\eta_{n,A,X}\in L^1(\R)$ depending only on $n,A,X$ such that
\begin{equation*}
  \cR_n (A,X)
=\int_{\mathbb{R}} f^{(n)}(t)\eta_n(t)\,dt
\end{equation*}
and
\begin{equation}
  \label{stima.eta}
\left\|\eta_n \right\|_{L^1}\leq c_n\left\| X\right\|_{S^n}^n,
\end{equation}
with $c_n$ which, for any $n$ is  absolute constant.
\end{theorem}

\begin{corollary}
  \label{C.33} Let $f$ be as above.
  For all $Y\in HS$ one has
  \begin{equation}
    \label{siffresto}
\left|d\cR_n(A,X)Y\right|\leq C\left\|f^{(n)}\right\|_{L^\infty}\left\|
X\right\|_{HS} ^{n-1}\left\| Y\right\|_{HS}\ .
    \end{equation}
\end{corollary}
\proof Using Eq.\eqref{eq:main-formula} one gets
\begin{equation*}
 d\mathcal R_n(A,X)[Y]
 =\frac{1}{(n-2)!}\int_0^1(1-t)^{n-2}\,
 d^nF(L_0+tX)[\rep{X}{n-1},Y] \,dt,
\end{equation*}
so we have to estimate the norm of $d^nF$. To this end 
fix a self-adjoint operator $Y\in S^n $ and a small real number $h$. Applying
Theorem \ref{th.PSS} with base operator $A_t$ and perturbation $hY$ gives
\begin{equation}
 \left|\cR_{n}(A_t,hY)\right|
 \leq C \left\| f^{(n)}\right\|_{L^\infty} |h|^n\|Y\|_{S^n}^n.
 \label{eq:PSS-at-At}
\end{equation}
Since $f$ is a smooth, the scalar function
\begin{equation*}
 h\longmapsto F(A_t+hY)
\end{equation*}
is smooth and thus the limit 
\begin{equation*}
 \lim_{h\to0}\frac{\cR_{n}(A_t,hY)}{h^n}
 =\frac{1}{n!}d^nF(A_t)\big[\rep{Y}{n}\big]
 \end{equation*}
exists.
Dividing \eqref{eq:PSS-at-At} by $|h|^n$ and passing to the limit yields
\begin{equation*}
 {
 \left|d^nF(A_t)[\rep{Y}{n}]\right|
 \leq C \left\| f\right\|_{L^\infty}\|Y\|_{S^n}^n,
 \qquad 0\leq t\leq1.
 }
\end{equation*}
The factor $n!$ has been absorbed into $C$. Finally the polarization
formula gives the result. 
\qed

Since we are interested in applying the theorem to the case where $f$
is a polynomial, which is an unbounded function, we have first to localize it to fulfill the
assumptions of Theorem \ref{th.PSS}. 

\begin{corollary}
  \label{cor.PSS}
Let $A$ and $X$ and $Y$ be selfadjoint matrices s.t. the spectrum of $A_t$ is
contained in a segment $[a,b]$. Let $\delta>0$ be a positive
parameter. 
Let $P\in C^\infty(\R)$ and denote
\begin{equation}
  \label{polu}
\cR_{P,n}(a,X):=\Tr \left( P(A+X)-\left.\sum_{k=0}^{n-1}\frac{1}{k!}
\frac{d^k}{dt^k}\bigl[P(A_t)\bigr]\right|_{t=0}\right)
  \end{equation}
then one has
\begin{align}
  \label{sti.tra.1}
\left|\cR_{P,n}(a,X)\right| 
\leq
C_\delta
\left(\sup_{l=0,...,,n}\sup_{x\in[a-\delta,b+\delta]}\left|P^{(l)}(x)\right| \right)
\left\|X\right\|_{S^n}^n\ 
\\
\left|d\cR_{P,n}(a,X)Y\right|  
\leq
C_\delta
\left(\sup_{l=0,...,,n}\sup_{x\in[a-\delta,b+\delta]}\left|P^{(l)}(x)\right| \right)
\left\|X\right\|_{HS}^{n-1}\left\| Y\right\|_{HS}\ .
\end{align}
\end{corollary}
\proof Let $\chi$ be a $C^\infty$ cutoff function which is equal to
$1$ in $[a,b]$ and has support in $[a-\delta,b+\delta]$, then Theorem
\ref{th.PSS} and Corollary \ref{C.33} apply to $f(x)=P(x)\chi(x)$,
giving the result. \qed

\section{Proof of two technical results.}\label{E}

\begin{lemma}
  \label{Ssl}
  Define
\begin{equation}
S_{s,l}:=\sum_{j=l}^{s}
\frac{(2s+1)!(2j+1)(j+l)!}
{(s-j)!(s+j+1)!(j-l)!}.
\end{equation}
Then one has 
\begin{equation}
{S_{s,l}=\frac{l!(2s+1)!}{s!(s-l)!}}.
\end{equation}
\end{lemma}

\proof Introduce the shifted Legendre polynomials
\begin{equation*}
Q_j(x):=P_j(2x-1),
\end{equation*}
where $P_j$ is the usual Legendre polynomial. Then Rodrigues formula,
see 8.4.6 of \cite{AS64} gives
\begin{equation}
  \label{rod}
Q_j(x)=\frac{1}{j!}\frac{d^j}{dx^j}\bigl[x^j(x-1)^j\bigr].
\end{equation}
Furthermore, by the standard Legendre orthogonality relation (see 22.1.5 of \cite{AS64}) one has
\begin{equation*}
\int_0^1 Q_i(x)Q_j(x)\,dx=\frac{\delta_{ij}}{2j+1}.
\end{equation*}

Since $x^s$ is a polynomial of degree $s$, it can be expanded on the basis $\{Q_0,\ldots,Q_s\}$:
\begin{equation}
  \label{2.6}
x^s=\sum_{j=0}^{s}a_{s,j}Q_j(x).
\end{equation}
By orthogonality,
\begin{equation*}
a_{s,j}=(2j+1)\int_0^1 x^sQ_j(x)\,dx.
\end{equation*}
Using the Rodrigues formula \eqref{rod},
\begin{equation*}
\int_0^1x^sQ_j(x)\,dx
=\frac{1}{j!}\int_0^1x^s\frac{d^j}{dx^j}
\bigl[x^j(x-1)^j\bigr]dx.
\end{equation*}
Integrating by parts $j$ times, all boundary terms vanish because $x^j(x-1)^j$ has zeros of order $j$ at both endpoints. Hence
\begin{equation*}
\int_0^1x^sQ_j(x)\,dx
=\frac{(-1)^j}{j!}\frac{s!}{(s-j)!}
\int_0^1x^{s-j}x^j(x-1)^j\,dx
=\frac{s!}{j!(s-j)!}\int_0^1x^s(1-x)^j\,dx.
\end{equation*}
By Euler's beta integral formula (see \cite{AS64} 6.2.1 and 6.2.2) one has
\begin{equation*}
\int_0^1x^s(1-x)^j\,dx
=\frac{s!j!}{(s+j+1)!},
\end{equation*}
and therefore
\begin{equation*}
\int_0^1x^sQ_j(x)\,dx
=\frac{(s!)^2}{(s-j)!(s+j+1)!}.
\end{equation*}
It follows that
\begin{equation*}
a_{s,j}=(2j+1)\frac{(s!)^2}{(s-j)!(s+j+1)!}.
\end{equation*}
Substitution into \eqref{2.6} gives
\begin{equation}
  \label{2.14}
x^s=\sum_{j=0}^{s}(2j+1)
\frac{(s!)^2}{(s-j)!(s+j+1)!}Q_j(x).
\end{equation}
We differentiate both sides $l$ times and evaluate at $x=1$. The left-hand side gives
\begin{equation*}
\left.\frac{d^l}{dx^l}x^s\right|_{x=1}
=\frac{s!}{(s-l)!}.
\end{equation*}
 To compute $Q_j^{(l)}(1)$ we make the substitution $x=1+y$ and use
 \eqref{rod} which takes the form
\begin{equation*}
Q_j(1+y)=\frac{1}{j!}\frac{d^j}{dy^j}
\bigl[(1+y)^jy^j\bigr].
\end{equation*}
Expanding the product,
\begin{equation*}
(1+y)^jy^j=\sum_{r=0}^{j}\binom{j}{r}y^{j+r}.
\end{equation*}
one can compute the Taylor expansion 
\begin{equation*}
Q_j(1+y)=\frac{1}{j!}\sum_{r=0}^{j}
\binom{j}{r}\frac{(j+r)!}{r!}y^r.
\end{equation*}
Therefore, for $j\geq l$, comparison of the coefficient of $y^l$ yields
\begin{equation*}
Q_j^{(l)}(1)=\frac{(j+l)!}{l!(j-l)!}.
\end{equation*}
For $j<l$, one has of course $Q_j^{(l)}(1)=0$. Inserting in (the
$l$-th differential of) \eqref{2.14}, we obtain
\begin{equation*}
\frac{s!}{(s-l)!}
=\sum_{j=l}^{s}(2j+1)
\frac{(s!)^2}{(s-j)!(s+j+1)!}
\frac{(j+l)!}{l!(j-l)!}.
\end{equation*}
Multiplying both sides by $l!(2s+1)!/(s!)^2$ gives
\begin{equation*}
\sum_{j=l}^{s}
\frac{(2s+1)!(2j+1)(j+l)!}
{(s-j)!(s+j+1)!(j-l)!}
=\frac{l!(2s+1)!}{s!(s-l)!}.
\end{equation*}
which proves the claim. \qed

\noindent{\it Proof of Lemma \ref{prop:main-estimates}.}
First, denote $a:=1+\delta$.
For $k=0,1,2$ one has
\begin{equation}
 f^{(k)}(z)=(-1)^k\frac{s!}{(s-k)!} \int_0^1
 \frac{t^k(1-tz)^{s-k}}{\sqrt{1-t}}\,dt.
\label{eq:f-k-derivatives}
\end{equation}
We first prove the following auxiliary bound: for $k=0,1,2$ there exists
$C_{k,a}>0$ such that, uniformly for $0\le z\le a^2$,
\begin{equation}
 \int_0^1 \frac{t^k |1-tz|^{s-k}}{\sqrt{1-t}}\,dt
 \le \frac{C_{k,a}}{(1+sz)^{k+1}}.
\label{eq:aux-bound}
\end{equation}

We split the integral into $[0,1/2]$ and $[1/2,1]$. On $0\le t\le 1/2$, since
$z\le a^2<2$, one has $0\le tz<1$. Hence
\begin{equation*}
 |1-tz|=1-tz\le e^{-tz}.
\end{equation*}
Thus, for $k=0,1,2$,
\begin{equation*}
 \int_0^{1/2}\frac{t^k |1-tz|^{s-k}}{\sqrt{1-t}}\,dt
 \le \sqrt{2}\int_0^{1/2} t^k e^{-(s-k)tz}\,dt
 \le \frac{C_k}{(1+sz)^{k+1}}.
\end{equation*}

On $1/2\le t\le 1$, set
\begin{equation}
 c_a:=\min\left\{\frac12,\frac{2-a^2}{a^2}\right\}>0.
\label{eq:def-ca}
\end{equation}
Then
\begin{equation*}
 |1-tz|\le 1-c_a z\leq e^{-c_az},
 \qquad 0\le z\le a^2,
 \quad \frac12\le t\le 1.
\end{equation*}
Indeed, $1-tz\le 1-z/2$, while $1-tz\ge 1-z$, and therefore
$|1-tz|\le \max\{1-z/2,z-1\}$. The choice of $c_a$ in \eqref{eq:def-ca}
implies that both terms are bounded by $1-c_a z$. Consequently,
\begin{equation*}
 |1-tz|^{s-k}\le e^{-c_a(s-k)z},
\end{equation*}
so that
\begin{equation*}
 \int_{1/2}^{1}\frac{t^k |1-tz|^{s-k}}{\sqrt{1-t}}\,dt
 \le 2 e^{-c_a(s-k)z}
 \le \frac{C_{k,a}}{(1+sz)^{k+1}}\ .
\end{equation*}
Then \eqref{eq:aux-bound} follows.

Using \eqref{eq:f-k-derivatives} and \eqref{eq:aux-bound}, we obtain the pointwise estimates
\begin{equation}
 |f(z)|\le \frac{C_a}{1+sz}\ ,
\label{eq:f-estimate}
\end{equation}
\begin{equation}
 |f'(z)|\le C_a\frac{s}{(1+sz)^2}\ ,
\label{eq:fprime-estimate}
\end{equation}
\begin{equation}
 |f''(z)|\le C_a\frac{s^2}{(1+sz)^3}\ ,
\label{eq:fsecond-estimate}
\end{equation}
valid for $0\le z\le a^2$.

Then exploiting the definition of $F$ one gets 
\begin{equation}
 |F(y)|\le C_a\frac{|y|}{1+s y^2}\ .
\label{eq:F-basic-bound}
\end{equation}
If $u=\sqrt{s}|y|$, then
\begin{equation*}
 \frac{|y|}{1+s y^2}=\frac{1}{\sqrt{s}}\frac{u}{1+u^2}
 \le \frac{C}{\sqrt{s}}\ ,
\end{equation*}
and this proves \eqref{eq:F-estimate}.

Similarly, using \eqref{eq:f-estimate} and \eqref{eq:fprime-estimate},
\begin{equation*}
 |F'(y)|=\left|f(y^2)+2y^2 f'(y^2)\right|\le C_a\left[\frac{1}{1+s y^2}
 +\frac{s y^2}{(1+s y^2)^2}\right].
\end{equation*}
The two factors in brackets are uniformly bounded for $y\in\mathbb R$ and
$s\ge 1$, hence \eqref{eq:Fprime-estimate} follows.

Finally, by \eqref{eq:fprime-estimate} and \eqref{eq:fsecond-estimate},
\begin{equation*}
 |F''(y)|=\left|6y f'(y^2)+4y^3 f''(y^2)\right|\le C_a\left[\frac{s|y|}{(1+s y^2)^2}
 +\frac{s^2|y|^3}{(1+s y^2)^3}\right].
\end{equation*}
With $u=\sqrt{s}|y|$, the two terms become
\begin{equation*}
 \frac{s|y|}{(1+s y^2)^2}
 =\sqrt{s}\frac{u}{(1+u^2)^2},
\quad  \frac{s^2|y|^3}{(1+s y^2)^3}
 =\sqrt{s}\frac{u^3}{(1+u^2)^3}.
\end{equation*}
Both are rational functions of $u$ are bounded on $[0,\infty)$, and hence
\eqref{eq:Fsecond-estimate} follows.

It remains to estimate the primitive. By \eqref{eq:F-basic-bound}, for
$y\in(-\delta,1+\delta)$,
\begin{equation*}
 |G(y)|\le C_a\int_0^{|y|}\frac{u}{1+s u^2}\,du.
\end{equation*}
Since $|y|\le a$, one has
\begin{equation}
 \int_0^{|y|}\frac{u}{1+s u^2}\,du
 =\frac{1}{2s}\log(1+s y^2)
 \le \frac{C_a\log(2+s)}{s}.
\label{eq:G-log-bound}
\end{equation}
This proves \eqref{eq:G-estimate}. 
\qed

\bibliographystyle{alpha}
\bibliography{Toda_vs_FPU.bib}

\end{document}